\documentclass[sigconf]{acmart}
\pdfoutput=1
\usepackage{tikz, xspace, amsmath, amssymb, boxedminipage,color,
	multirow, paralist,enumitem,fancyhdr, graphicx, caption,
	subcaption, amsthm}

\usepackage{tikz}
\usetikzlibrary{shapes,arrows}
\usepackage{microtype}

\usepackage{environ}
\usepackage[keys, probability, sets, logic, operators, adversary, ff, primitives, advantage,asymptotics]{cryptocode}

\createprocedurecommand{gamer}{}{}{linenumbering}
\NewEnviron{game}[1][]{\noindent\gamer{#1}{\BODY}}

\newcommand{\KYCS}{\ensuremath{\mathtt{KYChain}}\xspace}
\newcommand{\Setup}{\ensuremath{\mathtt{Setup}}\xspace}
\newcommand{\Register}{\ensuremath{\mathtt{RegisterU}}\xspace}
\newcommand{\Commit}{\ensuremath{\mathtt{SubmitU}}\xspace}
\newcommand{\Update}{\ensuremath{\mathtt{UpdateU}}\xspace}
\newcommand{\ProofVerify}{\ensuremath{\mathtt{Certify}}\xspace}
\newcommand{\Prove}{\ensuremath{\mathtt{CertifyU}}\xspace}
\newcommand{\Certify}{\ensuremath{\mathtt{CertifyC}}\xspace}
\newcommand{\Onboard}{\ensuremath{\mathtt{Verify}}\xspace}

\newcommand{\Time}{\ensuremath{\mathtt{Time}}\xspace}
\newcommand{\Append}{\ensuremath{\mathtt{Append}}\xspace}
\newcommand{\Search}{\ensuremath{\mathtt{Search}}\xspace}

\newcommand{\PRF}{\ensuremath{\mathtt{PRF}}\xspace}
\newcommand{\DS}{\ensuremath{\mathtt{DS}}\xspace}
\newcommand{\Sign}{\ensuremath{\mathtt{Sign}}\xspace}
\newcommand{\Verify}{\ensuremath{\mathtt{Vrfy}}\xspace}
\newcommand{\SE}{\ensuremath{\mathtt{SE}}\xspace}

\newcommand{\Oracle}{\ensuremath{\mathcal{O}}\xspace}
\newcommand{\Oreg}{\ensuremath{\mathtt{Oreg}}\xspace}
\newcommand{\Oadd}{\ensuremath{\mathtt{Oadd}}\xspace}

\newcommand{\Oupdate}{\ensuremath{\mathtt{Oupdate}}\xspace}
\newcommand{\Ocommit}{\ensuremath{\mathtt{Osubmit}}\xspace}
\newcommand{\Oprove}{\ensuremath{\mathtt{Oprf}}\xspace}
\newcommand{\Ocert}{\ensuremath{\mathtt{Ocert}}\xspace}
\newcommand{\Ocorrupt}{\ensuremath{\mathtt{Ocor}}\xspace}

\newcommand{\privacy}{\ensuremath{\mathtt{dc}}\xspace}
\newcommand{\authentication}{\ensuremath{\mathtt{cc}}\xspace}

\newcommand{\User}{\ensuremath{\mathtt{U}}\xspace}
\newcommand{\Ledger}{\ensuremath{\mathtt{PL}}\xspace}
\newcommand{\Bank}{\ensuremath{\mathtt{C}}\xspace}

\newcommand{\param}{\ensuremath{\mathtt{pp}}}
\newcommand{\interactive}[4]{\ensuremath{\langle #1 (#2),\allowbreak #3 (#4)\rangle}}

\newcommand{\squarebracket}[1]{\ensuremath{[#1]}}

\newcommand{\uid}{\ensuremath{\mathtt{uid}}\xspace}
\newcommand{\upi}{\ensuremath{\mathtt{upi}}\xspace}

\newcommand{\Doc}{\ensuremath{\mathtt{kyc}\mbox{-}\mathtt{data}}\xspace}
\newcommand{\type}{\ensuremath{\mathtt{type}}\xspace}
\newcommand{\data}{\ensuremath{\mathtt{data}}\xspace}
\newcommand{\desc}{\ensuremath{\mathtt{certL}}\xspace}
\newcommand{\acc}{\ensuremath{\mathtt{accL}}\xspace}
\newcommand{\auth}{\ensuremath{\mathtt{auth}}\xspace}

\newcommand{\cdata}{\ensuremath{\mathtt{cdata}}\xspace}
\newcommand{\cdesc}{\ensuremath{\mathtt{ccertL}}\xspace}
\newcommand{\cacc}{\ensuremath{\mathtt{caccL}}\xspace}
\newcommand{\recid}{\ensuremath{\mathtt{rid}}\xspace}
\newcommand{\etime}{\ensuremath{\mathtt{time}}\xspace}

\newcommand{\policy}{\ensuremath{\Psi}\xspace}
\newcommand{\bool}{\ensuremath{\mathtt{bool}}\xspace}
\newcommand{\certificate}{\ensuremath{\mathtt{cert}}\xspace}

\newcommand{\query}{\ensuremath{\mathtt{query}}\xspace}
\newcommand{\entrylist}{\ensuremath{\mathtt{Tlist}}\xspace}

\newcommand{\abort}{\ensuremath{\mathtt{abort}}\xspace}

\newcommand{\seed}{\ensuremath{\mathtt{seed}}\xspace}

\newcommand{\Log}{\ensuremath{\mathtt{Log}}\xspace}
\newcommand{\RegLog}{\ensuremath{\mathtt{RLog}}\xspace}

\newcommand{\Entry}{\ensuremath{\mathtt{rec}}\xspace}
\newcommand{\Tran}{\ensuremath{\mathtt{t}\mbox{-}\mathtt{rec}}\xspace}

\newcommand{\Exp}{\ensuremath{\mathtt{G}}\xspace}

\newtheorem{remark}{Remark}

\renewcommand\footnotetextcopyrightpermission[1]{} 
\PassOptionsToPackage{hyphens}{url}\usepackage{hyperref}
\begin{document}
\title{KYChain: User-Controlled KYC Data Sharing and Certification}
	
\author{Constantin C\u{a}t\u{a}lin Dr\u{a}gan}
\email{c.dragan@surrey.ac.uk}
\affiliation{%
	\institution{Surrey Centre for Cyber Security\\ 
		University of Surrey\\Guildford, United Kingdom }	
}

\author{Mark Manulis}
\email{m.manulis@surrey.ac.uk}
\affiliation{%
\institution{Surrey Centre for Cyber Security\\ 
	University of Surrey\\Guildford, United Kingdom }	
}

\begin{abstract}
Under Know Your Customer (KYC) regulations, financial institutions are required to verify the identity and assess the trustworthiness of any new client during on-boarding, and maintain up-to-date records for risk management. These processes are time consuming, expensive, typically have sub-par record-keeping steps, and disadvantage clients with nomad lifestyle. In this paper, we introduce KYChain as a privacy-preserving certification mechanism that allows users to share (certified) up-to-date KYC data across multiple financial institutions. We base KYChain on immutable ledgers and show that it offers confidentiality and certification compliance of KYC data.

%

\end{abstract}

\keywords{Know Your Customer, Privacy-Preserving, Distributed Ledger Technology, Certification}

\maketitle

\section{Introduction}

Know your customer, or simply KYC, is a regulated process \cite{EU-KYC, UK-KYC} requiring financial institutions (FIs, e.g., banks) to verify identities and check transactional behaviours of their clients to facilitate detection of suspicious activities (e.g. money laundering). Typical implementations of KYC compliance require customers to provide due dilligence information to their FIs, starting with initial personal information during the on-boarding stage and reporting any subsequent updates while their business relationship exists. Recent studies show that an on-boarding process can take up to 32 days/customer, greatly impacting the overall KYC compliance costs, which can be up to \$20k/year \cite{Marous} for each new client. These costs are then passed to customers in the form of high transaction fees. Moreover, inadequate handling of KYC data (e.g. duplicate or confusing requests, and lack of common KYC standards from different FIs) have lead to 12\% of corporate clients changing their FI in 2017 \cite{Thomsonreuters}. 
%
%
There is a spread of commercial KYC services (e.g., Trulioo, Pegasystems, LexiNexis, Deloitte KYC Start) offered by companies that operate on customer's data and assist FIs in the verification process. These solutions do not allow re-use of KYC data across multiple FIs and more importantly do not provide users with full control over their KYC data, a key requirement behind recent GDPR regulations. More recently, some commercial KYC services (e.g., Coinfirm, Tradle, KYC Legal) have adopted blockchain technologies and proprietary mechanisms to facilitate secure exchange of due diligence information between multiple FIs. These services must still be trusted with confidentiality of the customer's KYC data. 
The few existing academic approaches focus either on re-using certified KYC data from one FI to another without re-certification, but do not provide confidentiality 
\cite{ParraMoyano2017}, or by sharing the KYC data in a private-preserving manner with FIs that perform their own certification \cite{DoubleBlind}.

\paragraph{Contribution} We propose KYChain, a privacy-preserving certification protocol that enables secure sharing of up-to-date KYC data across multiple FIs and is fully controlled by the clients. At the core of KYChain is an immutable ledger that stores hashes of (encrypted) KYC data and certificates that are issued for a particular customer. The corresponding ciphertexts encrypting KYC data and certificates obtained from other FIs are encrypted and stored in an off-chain storage. The client keeps decryption keys, which can be issued to FIs with whom the client wishes to establish or maintain a business relationship. The immutable ledger helps to keep track of all user-submitted changes for the KYC data. FIs can monitor the ledger to identify which KYC data has been updated and request keys from the customer to obtain these updates. In contrast to \cite{ParraMoyano2017, DoubleBlind}, KYChain can help to reduce the on-boarding time of new clients through the possibility of reusing (certified) KYC data across multiple FIs. This would reduce costs associated with the ongoing monitoring of changes to KYC data by the FIs through automated detection of updates on the ledger and off-chain storage.  KYChain can be offered as a service by an entity who would be running the off-chain storage without jeopardizing confidentiality of customer's data. Furthermore, we define the security properties that enforce guarantees over confidentiality of the KYC data and authenticity of  certification for KYC data compliance.

\section{KYChain Model and Requirements}

\subsection{Entities: Clients, Ledger, Certifiers}

\paragraph{Clients}
We model KYC clients by their unique personal identifier \upi, e.g., name and personal numeric code or passport number. This unique \upi is used upon registration to assess the actual identity of the client and avoid fraudulent registrations. Clients then generate their own private/public key pair, and can use their public key as a cryptographic identity in the system. Moreover, we allow clients to register multiple public keys as long as they are linked to their \upi. This registration is handled by trusted certifiers who keep log of matched public keys and \upi.

	
\paragraph{Public Ledger} 
KYChain adopts a distributed public ledger with an assumed off-chain storage for the records. Clients will store their KYC related data and obtained certificates in an encrypted form off-chain with corresponding hashes committed into the ledger to guarantee integrity. For simplicity, we model this ledger/storage combination as a single entity that realises an {\em append-only list} and adds a timestamp to each record it receives. 
Additionally, we assume that search queries can be performed over the ledger and the off-chain storage based on timestamps and the information contained in committed records. Formally, we define the public ledger \Ledger=(\Setup, \Time, \Append, \Search) with the following algorithms:
\begin{description}
	\item[\Setup($\lambda$):] \param. Initializes the append-only list, starts the tamper-proof clock, and returns the public parameters \param;
	
	\item[\Time(\;):] \etime. Returns the current time from the internal clock;
	
	\item[\Append(\Entry):] $\Tran \cup \{\bot\}$. Returns either a valid timestamped record \Tran= (\etime, \Entry) for the input \Entry received at time $\etime \gets \Time()$; otherwise an error symbol $\bot$;
	
	\item[\Search(\query):] \entrylist. Returns a list of timestamped records \entrylist that satisfy the search requirements in \query.
\end{description}

\paragraph{Public Ledger with external database}
For ease of description, we refer to the public ledger as a single entity \Ledger. However, we consider a hybrid approach for the public ledger instantiation, with an external database $\mathtt{DB}$  for storage and a public blockchain $\mathtt{BL}$ for integrity. More precisely, for each record \Entry submitted by a client, the timestamped record \Tran= $(\etime,\Entry)$ is first recorded by $\mathtt{DB}$ and its hash $\hash(\Tran)$ is then committed into the blockchain $\mathtt{BL}$. This approach allows clients and certifiers access to a search functionality, performed over the database $\mathtt{DB}$, and extract timestamped records. Furthermore, only clients can update or remove their KYC data from the database. 
We can use {\em Ethereum} and {\em Bitcoin} as existing implementation for our blockchain. More details on the setup are provided in Section \ref{sec: constr}. 


\paragraph{Certifiers}
Clients commit into the ledger some digital representation of their KYC data, e.g., scans of passports, ID cards, utility bills, photographs, etc. A certifier interacts with the client to verify that their digital information matches the clients' personal identifier \upi. If this check is successful, the certifier produces a certificate attesting to the correctness of the client's KYC data. 
We assume that each certifier has some policy \policy defined over KYC data and that certificates are issued only if client's KYC data $A$ satisfies the certifier's policy, i.e., $\policy(A)= 1$. We assume that all eligible certifiers are trusted and publicly known to all parties in the system. 

\subsection{KYC Data and Certificates}

\paragraph{KYC Data} 
We follow a specific template when modeling the KYC data:
\[ 
\Doc = (\pk,\type,\data,\desc,\acc)
\]
\begin{itemize}
\item $\pk$ is the public key of the client that submits this KYC data;

\item \type describes what type of KYC data it is, e.g., {\em passport, id card, location, occupation, bills, etc. };

\item \data is a digital copy of the KYC data;	

\item \desc 
is a list of all certificates issued for this KYC data. 

\item \acc 
enumerates all certifiers that can access the client's KYC data. The list starts empty, and then gets updated by the client. 

\end{itemize}

%
%
%

\paragraph{Linking KYC Data with the Ledger}
KYChain does not process private information of the clients. Given some \Doc = (\pk,\type,\allowbreak\data,\desc,\acc), the client first encrypts its KYC data, certificates and certifiers, and adds an authenticator to prove the origin. This results in the KYC data record of the form:
\begin{equation}\label{eq: entry}
\Entry = (\pk, \recid, \type,\cdata,\cdesc,\cacc, \auth)
\end{equation}
\begin{itemize}
	
	\item \recid is a unique identifier for this record; 
	
	\item \cdata is a ciphertext resulted by encrypting the value \data; 
	
	\item \cacc is a list of ciphertexts resulted from encrypting the certifiers from \acc;
	
	\item \cdesc is a list of ciphertexts formed by encrypting the obtained certificates from \desc;
	
	\item \auth is an authenticator over (\type, \cdata, \cdesc, \cacc), that can be publicly verified using client's public key \pk.
	
\end{itemize}

The ledger stores timestamped records built on the records clients submit. For simplicity, we adapt the notation \Tran=(\etime, \Entry) to  
\begin{equation}\label{eq: transaction}
\Tran = (\pk, \etime, \recid, \type,\cdata,\cdesc,\cacc, \auth)
\end{equation}
\begin{itemize}
	\item \etime is the time the record has been received.
	
\end{itemize}


\subsection{KYChain Certification Scheme: Definition}
The core of KYChain are rigorous protocols for establishing the identity of potential clients, measuring their degree of trustworthiness, and continued monitoring for risk assessment. 
We mirror the on-boarding process of a client with some FI by subsequent registration and certification of their KYC data performed by the FI. In addition, KYChain introduces a verification mechanism that allows clients to authenticate and share their certified KYC data with other FIs to speed up the eventual on-boarding process with them. Through the use of ledgers that store encrypted KYC data and certificates, previously authorized certifiers would be able to monitor changes to the client's KYC data and obtain updates using the same authorization mechanism as in the on-boarding phase. 

%
%

\def\defprivacy{
\procedure{$Exp_{\adv,\KYCS}^{\privacy,\beta}(\upi, \lambda)$}{
\pcln \param \gets \Setup (\lambda);\;\;  (\sk^*,\pk^*) \gets \Register(\param,\upi)\\
\pcln  \Log \gets \{(\sk^*,\pk^*, \upi)\}\\
\pcln (\type,\data_0,\data_1,\desc,\acc) \gets
      \adv_1^{\Oracle}(\param,\pk^*) \\	
\pcln  \Tran \gets \Commit(\param, \sk^*, \type, 
       \data_{\beta}, \desc, \acc)\\
\pcln \beta' \gets \adv_2^{\Oracle} (\Tran) \\	
\pcln \pcreturn (\beta' = \beta) \land \adv \mbox{ did not call } \Oprove(\pk^*, \pk_{\Bank},\Psi)	\mbox{ with } \\
\pcind[3]  \Tran \in \Search(\policy,\pk^*) \land \adv \mbox { did not call }   \Ocorrupt(\pk^*) 
}
}


\newsavebox\PmatrixO
\setbox\PmatrixO\hbox
{
	\scalebox{0.85}{
	$\begin{array}{l}
	\adv \mbox{ did not call } \Ocorrupt(\pk^*)\; \lor \\
	\lnot \Psi(A) \mbox { for } 
	A =\{ \Tran=(\pk^*,\cdot) 
	| \mbox{ for } \Tran \mbox{ added by } \Ocommit \mbox{ to } \Ledger \}		
	\end{array}$ 
	}
}

\def\defsound{
\procedure{$Exp_{\adv,\KYCS}^{\authentication}(\upi,\lambda)$}{
\pcln \param \gets \Setup (\lambda);\;\;  (\sk^*,\pk^*) \gets \Register(\param,\upi)\\
\pcln  \Log \gets \{(\sk^*,\pk^*, \upi)\}\\
\pcln \adv^{\Oracle,\Ocert}(\param,\pk^*) \\
\pcln\label{cc:return} 
   \hspace*{-2.5mm} \scalebox{0.80}{4.1}:\hspace{0.5mm} \pcreturn \; \adv \mbox{ made call to } \Ocert(\cdot) \mbox{ with } 
         \certificate \gets \Certify(\param,\pk^*, \cdot,\policy) \land \\
   \hspace*{3.5mm} \scalebox{0.80}{4.2}: \hspace{1.0mm} \true \gets \Onboard(\param,\pk^*,\policy,\certificate) \land \lnot \policy(\emptyset) \land 
         \adv \mbox{ is not running } \Oprove(\cdot)\; \land \\
   \hspace*{3.5mm} \scalebox{0.80}{4.3}:\hspace{1.0mm} \left(\usebox\PmatrixO\right)	 		
}}

\begin{figure*}[t!]
	\centering
	\begin{subfigure}[b]{0.48\textwidth}
		\centering
		\defprivacy
		\vspace*{0.0mm}
		\caption{Data Confidentiality with $\adv =(\adv_1,\adv_2)$.}
		\label{fig: conf}
	\end{subfigure}%
	~ 
	\begin{subfigure}[b]{0.48\textwidth}
		\centering
		\defsound
		\vspace*{1.0mm}
		\caption{Certification compliance. }
		\label{fig: auth}
	\end{subfigure}
	\caption{Security Properties, for \Oracle = \{\Oreg, \Ocorrupt, \Ocommit, \Oupdate, \Oprove\}.}
	\label{fig: sec}
\end{figure*}

\begin{definition}[KYChain Certification and Data Sharing]
The protocol \KYCS(\Ledger) = (\Setup, \allowbreak \Register, \allowbreak \Commit, \allowbreak \ProofVerify, \allowbreak \Onboard) has access to the ledger \Ledger, and consists of the following algorithms:
	
\begin{description}
\item[\Setup($\lambda$):] \param. Initializes the ledger \Ledger by calling its setup algorithm, defines the list of certifiers, and publishes the public parameters of the protocol \param.
	
\item[\Register(\param, \upi):] $  (\pk,\sk).$ 
Client \User generates locally a public-secret key pair (\pk,\sk), then he submits to certifier \Bank the public key $\pk$ with his personal identifier \upi. Certifier \Bank validates \upi, then stores (\pk,\upi).
	
\item[\Commit(\param,\sk, \Doc):] \Tran.  
Client \User(\pk,\sk) builds a record \Entry as described in Equation (\ref{eq: entry}) and 
calls \Ledger.\Append(\Entry). Ledger \Ledger verifies that \pk\ is registered, and that \auth is valid w.r.t \pk. After a successful verification, \Ledger computes \Tran according to Equation (\ref{eq: transaction}), with $\etime \gets \Ledger.\Time(\;)$. Finally, \Ledger stores \Tran locally, before sending a copy to \User. 

	
\item[\Update(\param, \sk,\recid,\squarebracket{\data'}, \squarebracket{\acc'}, \squarebracket{\desc'}):] $\Tran'.$
Client \User (\pk, \sk) uses this algorithm to update the timestamped record \Tran indexed by \recid with one or more of the following: $\data', \acc', \mathtt{desc}'$. Client \User retrieves \Tran from \Ledger, extracts the initial \Doc, and updates it to $\Doc'$ with changed data. 
Then, it performs $\Tran' \gets \Commit (\param,\allowbreak \sk,\allowbreak \Doc')$. 


\item[$\ProofVerify(\param,\sk_{\User},\sk_{\Bank}, \policy):$]
    \interactive{\User} {\certificate}{\Bank}{\certificate}.  
It is an interactive algorithm run between client $\User(\pk_{\User},\sk_{\User})$ and the certifier $\Bank(\pk_{\Bank}, \allowbreak \sk_{\Bank})$, with \Bank establishing policy \policy. Both parties have access to the information stored in \Ledger. 
\begin{itemize}
	\item $\Prove(\param,\sk_{\User}, \pk_{\Bank}, \policy)$ is run by client \User by interacting with \Certify to authenticate and show compliance with the policy \policy. The algorithm returns either a valid certificate \certificate, or abort with $\bot$;
	
	\item $\Certify(\param, \sk_{\Bank}, \pk_{\User}, \policy)$ run by certifier \Bank s.t. interacting with an authenticated and policy compliant \User it produces a certificate \certificate for this client; otherwise aborts with $\bot$;
\end{itemize}



\item[$\Onboard(\param, \pk_{\User}, A, \policy, \certificate):$] \bool. Run by any party that has access to the certificate \certificate, and used to verify that client 
$\User(\pk_{\User}, \cdot)$ has a valid certificate \certificate over policy \policy that is satisfied by set $A$ of KYC data. 
Typically, the verifier obtaines the certificate and KYC data following an authorisation from the client. 
\end{description}

\end{definition}

\subsection{Security Properties}
As security guarantees for our protocol, we focus on {\em data confidentiality} for the client's KYC data, and {\em certification compliance}.

\paragraph{Oracles}
For our experiments we consider that the adversary can register multiple clients, but is challenged on a single client $\User(\pk^*,\sk^*)$ generated by the experiment. 
The adversary can directly interact with the ledger \Ledger, and call all algorithms offered by \Ledger with the exception of \Setup, therefore, we do not model them as oracles. Moreover, the verification can be performed by anyone. Given the protocol \KYCS, we have the following list of oracles, that an adversary can access:
\begin{itemize}
\item \Oreg (\upi): \pk. Calls $(\pk,\sk) \gets \Register(\param,\upi)$, stores (\pk,\sk,\upi) internally in \Log, and returns \pk.
\item \Ocorrupt(\pk): \sk. Finds in $(\pk,\sk,\upi) \in \Log$, and returns \sk.

\item \Ocommit(\pk,\Doc ): $\Tran \cup \{\bot\}$. Finds $(\pk,\sk,\upi) \in \Log$, and returns $\Tran \gets \Commit(\param, \allowbreak \sk, \allowbreak \Doc)$. Otherwise, it returns the error symbol $\bot$. 
\item \Oupdate(\pk,\recid,$\cdot $): $\Tran'\cup \{\bot\}$. It searches for $(\pk,\sk,\upi) \in \Log$, and if no such entry is found it returns $\bot$. Otherwise, it returns $\Tran' \gets$ \Update(\param, $\sk$, \recid, $\cdot$).

\item $\Oprove(\pk_{\User}, \pk_{\Bank}, \policy)$: \certificate. Both $\pk_{\User}$ and $\pk_{\Bank}$ have to be in \Log. 
The adversary plays the role of a malicious certifier $\pk_{\Bank}$ by interacting with 
$\Prove(\param, \pk_{\User}, \policy)$.

\item $\Ocert(\pk_{\User},\pk_{\Bank},\policy): \certificate$. The adversary plays the role of a malicious client $\pk_{\User}$ by interacting with $\Certify(\param, \sk_{\Bank},\allowbreak \pk_{\User},\allowbreak \policy)$. Preliminarily, both $\pk_{\User}$ and $\pk_{\Bank}$ are verified to be registered.
   
\end{itemize}

\paragraph{Data Confidentiality}
Intuitively, the timestamped records in the ledger should not leak information about their data with the exception of the meta-information, i.e, type, time and public key. We model this property using a PPT adversary $\adv=(\adv_1,\adv_2)$ that needs to distinguish between two different KYC data by seeing a timestamped record for one of them in the ledger. More precisely, for a client $(\sk^*,\pk^*,\upi)$ the adversary selects 2 different KYC data: $\data_0$ and $\data_1$, and receives the timestamped record of $\data_\beta$ for a fixed bit $\beta\in \{0,1\}$ unknown to the adversary. The adversary has access to the functionalities of \KYCS: register, submit, update and prove KYC data (via oracles \Oreg,\Oadd, \Oupdate, \Oprove), together with the ability to corrupt (by calling oracle \Ocorrupt). The adversary wins if he can provide a guess $\beta'$ such that $\beta =\beta'$ under the condition he didn't ask for the secret key of the client $(\pk^*,\sk^*,\upi)$ (using the corruption oracle \Ocorrupt) and didn't ask ask for a decryption of the timestamped record \Tran (using the proof oracle \Oprove). We formalize this property in Figure \ref{fig: conf}. This property can be extended to show confidentiality for the access list \acc, and credential list \desc.

\begin{definition}[Data Confidentiality]
$\KYCS$ satisfies data confidentiality, if for any PPT adversary $\adv$ and any \upi, the following advantage is negligible in $\lambda$:
\[
\advantage{\privacy}{\adv,\KYCS}[] = 
\abs{\prob{Exp_{\adv,\KYCS}^{\privacy,\beta}(\upi,\lambda) =1} - \frac{1}{2}}.
\]	
\end{definition}
	
\paragraph{Certification Compliance}
Honest certifiers would only be able to create certificates for authenticated clients that satisfy their policy. In Figure \ref{fig: auth}, we model a PPT adversary \adv\ that needs to convince a certifier to create a valid certificate (that can be verified) when he doesn't satisfy the policy, or when he impersonated another client or certifier. The adversary is given access to registration, submission, update and prove \KYCS functions (via oracles \Oreg,\Oadd, \Oupdate, \Oprove, \Ocert), and the capability to corrupt (with oracle \Ocorrupt). We exclude trivial policies, \policy($\emptyset$)=1, and restrict the adversary not to run \Ocert and \Oprove simultaneously.

\def\defreg{
\procedure{\Register(\param, \upi)}{
\pcln (\pk, \mathtt{sigk}) \gets \DS.\kgen(\lambda)\\
\pcln \seed \sample \{0,1\}^\lambda\\
\pcln \sk = (\pk,\mathtt{sigk},\seed)\\
\pcln \RegLog \gets \RegLog \cup \{(\pk,\uid)\}\\
\pcln \pcreturn (\pk, \sk)	
}}

\def\defver{
\procedure{$\Onboard(\param, \pk_{\User}, A, \policy, \certificate)$}{
\pcln \pcfor (\pk_{\User}, \etime, \recid, \type,\cdata,\\
	\pcind[4] \desc,\cacc, \auth) \in A\\
     \pcind[3] \pcif \lnot \DS.\Verify(\pk_{\User}, (\type,\cdata,\\
          \pcind[5] \desc,\cacc), \auth)\\
     \pcind[3] \pcthen\ \abort\\	
\pcln \pcreturn \DS.\Verify(\pk_{\User}, A, \certificate)\; \land\\
\pcind[5] \policy(A)			
}}

\def\defsub{
\procedure{\Commit(\param,\sk, \Doc)}{
\pcln (\pk, \mathtt{sigk},\seed) \gets \sk\\
\pcln (\pk,\type, \data, \desc, \acc ) \gets \Doc\\	
\pcln \pcif \lnot (\pk, \cdot) \in \RegLog \pcthen\ \pcreturn \bot\\
\pcln \recid \sample \{0,1\}^\lambda\\
\pcln\label{submit:keys} \pcfor i \in \{1,2,3\} \pcdo 
       \pcind[1] \key_i \gets \PRF_{\seed} (\pk,\allowbreak \recid, i)\\
\pcln \cdata \gets \SE.\enc(\key_1, \data)\\
\pcln \cdesc \gets \SE.\enc(\key_2, \desc)\\
\pcln \cacc \gets \SE.\enc(\key_3, \acc)\\
\pcln \auth \gets \DS.\Sign(\mathtt{sigk},\\
    \pcind[6] (\type,\cdata,\desc,\cacc))\\
\pcln \Entry \gets (\pk, \recid, \type, \cdata, \desc,\\
   \pcind[6] \cacc, \auth)\\
\pcln \Tran \gets \Ledger.\Append(\Entry)\\
\pcln \pcreturn \Tran
}}

\def\defupdate{
\procedure{\Update(\param, \sk,\recid,\squarebracket{\data'}, \squarebracket{\acc'}, \squarebracket{\mathtt{desc}'})}{
\pcln (\pk,\mathtt{sigk},\seed) \gets \sk\\
\pcln \Tran \gets \Ledger.\Search(\recid)\\
\pcln \pcif \lnot (\pk, \cdot) \in \RegLog \lor \Tran = \bot \pcthen\ \pcreturn \bot\\
\pcln (\pk, \etime, \recid, \type, \cdata, \\
    \pcind[5] \desc, \cacc, \auth) \gets \Tran\\
\pcln \pcfor i \in \{1,2,3\} \pcdo 
\pcind[1] \key_i \gets \PRF_{\seed} (\pk,\allowbreak \recid, i)\\
\pcln \pcif \data' \neq \bot \pcthen\ \data \gets \data'\\
    \pcind[3] \pcelse \data \gets \SE.\dec(\key_1, \cdata)\\
\pcln \pcif \desc' \neq \bot \pcthen\ \desc \gets \desc'\\
    \pcind[3] \pcelse \desc \gets \SE.\dec(\key_2, \desc)\\    
\pcln \pcif \acc' \neq \bot \pcthen\ \acc \gets \acc'\\
    \pcind[3] \pcelse \acc \gets \SE.\dec(\key_3, \cacc)\\
\pcln \Doc' \gets (\pk, \recid, \type, \data, \\
   \pcind[8] \desc, \acc )\\	
\pcln \Tran \gets \Commit(\param, \sk, \Doc')\\	
\pcln \pcreturn \Tran	
}}

\def\defproofTree{
\pseudocode{%
\Prove(\param,\sk_{\User} , \pk_{\Bank} \policy) \mbox { with } \sk_{\User}=(\pk_{\User},\mathtt{sigk}_{\User},\seed_{\User})
	\>\> \Certify (\param,\pk_{\User}, \sk_{\Bank} \policy) \mbox { with } \sk_{\Bank}=(\pk_{\Bank},\mathtt{sigk}_{\Bank},\seed_{\Bank}) \\ 
[0.1\baselineskip][\hline] \>\> \\ [-0.5\baselineskip] 
\pcln K \gets \emptyset; \hspace*{2mm} L \gets \Search(\policy, \pk_{\User}) \>\> A \gets \emptyset; \hspace*{2mm} L \gets \Search(\policy, \pk_{\User}) 
    \> \pclnr\\
\pcln \pcfor (\pk_{\User}, \etime, \recid, \type,\cdata,\desc,\cacc, \auth) \in L
    \>\> \pcfor (\pk_{\User}, \etime, \recid, \type,\cdata,\desc,\cacc, \auth) \in L \>\pclnr\\
\pcln \pcind \pcif \lnot \DS.\Verify(\pk_{\User}, (\type,\cdata,\desc,\cacc), \auth) 
\>\> \pcind \pcif \lnot \DS.\Verify(\pk_{\User}, (\type,\cdata,\desc,\cacc), \auth)  \> \pclnr \\
\pcln \pcind \hspace*{-1mm} \pcthen\ \abort
\>\>  \pcind \hspace*{-1mm} \pcthen\  \abort  \> \pclnr \\
\pcln \pcind \pcelse \\
\pcind[5] \key_1 \gets \PRF_{\seed_{\User}}(\pk_{\User},\recid,1) \\
\pcind[5] K \gets K \cup \{\key_1\} \>\> \\
\pcln \>\sendmessageright*[8mm]{K} \\ 
    \>\> \pcfor 0 \leq i \leq |L| \pcdo \> \pclnr\\
    \>\> \pcind (\pk_{\User}, \etime, \recid, \type,\cdata,\desc,\cacc, \auth) \gets L[i] \>\pclnr\\
    \>\> \pcind A \gets A \cup \{(\type, \etime, \SE.\dec(K[i], \cdata) ) \> \pclnr\\ 
    \>\> \pcif \lnot \policy(A) \pcthen\ \abort \> \pclnr\\
      \>\> \certificate \gets \DS.\Sign(\sk_{\Bank}, A) \>\pclnr \\
\pcln \pcreturn \certificate \>\sendmessageleft*[8mm]{\certificate}    \> \pcreturn \certificate	\>\pclnr	
}}

\begin{figure*}[!ht]
	\begin{pchstack}
		\begin{pcvstack}
			\defreg \pcvspace \defver
		\end{pcvstack}
		\pchspace \defsub \pchspace \defupdate
	\end{pchstack}
	\caption{The \Register, \Commit, \Update and \Onboard algorithms. }
	\label{fig: Reg-Sub-Upd}
\end{figure*}

\begin{definition}[Certification Compliance]
$\KYCS$ ensures certification compliance, if for any PPT adversary $\adv$ the following advantage is negligible in $\lambda$:
$$ 
\advantage{\authentication}{\adv,\KYCS}[] = 
	\prob{Exp_{\adv,\KYCS}^{\authentication}(\lambda) =1}.$$	
\end{definition}

\section{KYChain Scheme: Our Construction}

%

\subsection{Cryptographic Building Blocks}
Our system relies on standard cryptographic primitives that have classic security properties. We employ {\em pseudo-random functions} $\PRF: \{0,1\}^{\lambda} \times \{0,1\}^{\star}\to \{0,1\}^{poly(\lambda)}$ \cite{DBLP:journals/jacm/GoldreichGM86}, {\em digital signature scheme} \DS = (\kgen, \Sign, \Verify) \cite{DBLP:journals/tit/DiffieH76} that are existentially unforgeable under chosen message attack (EUF-CMA), and a {\em symmetric encryption scheme} \SE =(\kgen,\enc, \dec) with two security requirements: {\em indistinguishability under chosen plaintext attack (IND-CPA)}, and {\em wrong-key detection (WKD)} \cite{DBLP:conf/tcc/CanettiKVW10}. 
The WKD property states that any efficient adversary cannot successfully decrypt a ciphertext with a key different from the one used to encrypt it: $\dec(k',\enc(k,m)) = \bot$ for $k \neq k'$.  

\subsection{Generic KYChain Scheme}\label{sec: constr}

We provide an overview of our construction and how it integrates with KYC. We formalize the \Register, \Commit, \Update, \Onboard algorithms in Figure \ref{fig: Reg-Sub-Upd}, and the interactive algorithm \ProofVerify in Figure \ref{fig: constr}.

The setup phase is initialized by a trusted third party that starts the ledger \Ledger and defines the list of certifiers, accessible to all. 
{\color{red}} We consider a hybrid approach for the public ledger instantiation: external database for storage and a public blockchain for integrity. We can use existing public run blockchains, i.e., Bitcoin or Ethereum, and have each certifier \Bank maintain its own local database $\mathtt{DB}_\Bank$. Any record \Entry submitted by a client to a certifier, with be submitted as a timestamped record \Tran in its local database $\mathtt{DB}_\Bank$. Then, the certifier would send the timestamped record as a transaction to the public blockchain. In this scenario the certifier would be paying the transaction fee, associated to that transaction. Certifiers can optimize this process, and collect records from multiple clients received in a single day, and create a single transaction for all of them. 

\begin{remark}[Alternative Public Ledger Setup]
The above setup method has the advantage of being easily deployable with current FI infrastructure. However, an alternative would be to deploy a permissioned blockchain, with the certifiers acting as nodes. That is, the certifiers would be trusted to submit transactions to the rest of network. This would remove the need for transaction fees, but would require different trust assumptions, as certifiers would have the power to block or alter transactions. More likely, rigorous auditing mechanism would be needed to ensure certifiers do not or have not deviated from the protocol steps.  	
\end{remark}

To ease description we use a single \RegLog list with the public keys and personal identification number of all registered clients, that only certifiers can access. When joining the system, clients generate their own verification-signing keys $(\pk,\mathtt{sigk})$ together with a value \seed used to derive unique symmetric keys used to encrypt each KYC data. Clients register by submitting their public key \pk\ together with their personal identity \upi to the certifier $\pk_{\Bank}$. Once a client's identity is verified the \RegLog is updated  with (\pk,\upi). 

In practice, each certifier $\Bank(\pk_{\Bank},\sk_{\Bank})$ would have their own list $\RegLog_{\Bank}$ and would only share the public keys they have registered, together with personal identification of known untrustworthy clients or countries, i.e. black/grey lists \cite{FATF, BanksBlacklist,BanksBlacklistWSJ}. In case of honest clients no personal information is shared between different FIs. However, by sharing personal identifier of untrustworthy clients any FI would be able to block accounts and request resolution for their own clients deemed untrustworthy by other FIs.

\begin{figure*}[ht]
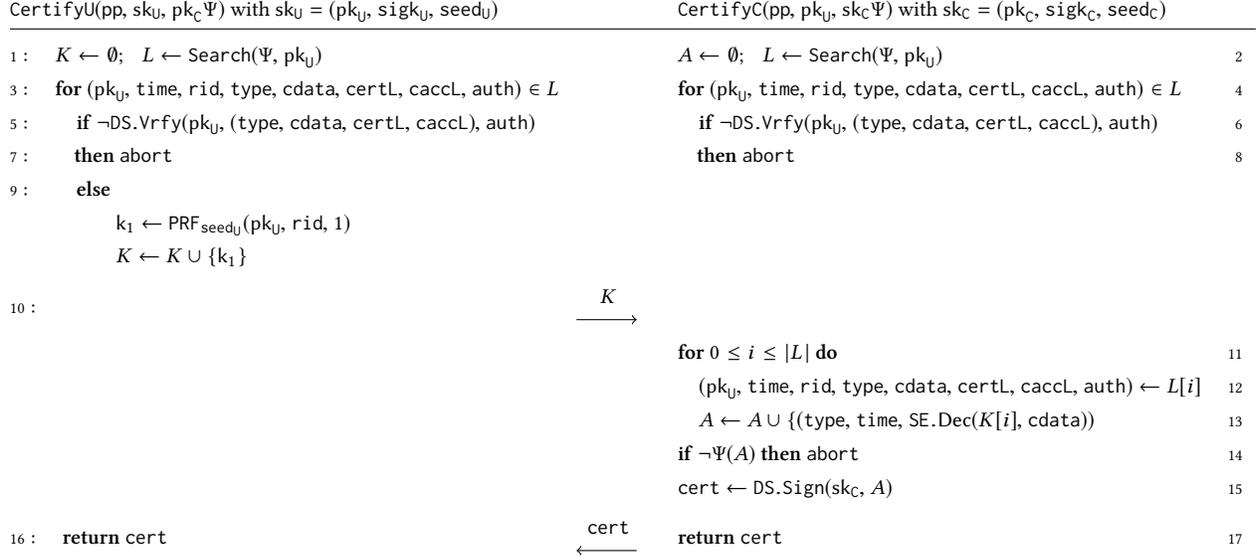

	\begin{pcvstack}
		\defproofTree
	\end{pcvstack}
	\caption{The \ProofVerify interactive algorithm. A preliminary step consists of verifying that both parties are registered. All communications are performed over an authenticated and confidential channel.}
	\label{fig: constr}
\end{figure*}

\begin{remark}[Offline and Online registration]
{\em Online registration} requires clients to fill online forms that contain their phone number, full name, current living address. This is coupled with a scan of their valid identity card or passport, and a recording that clearly shows their face. Certifiers check the validity of KYC data via communication with competent authorities, e.g. police, and verify that the person recorded matches the person on the scan KYC data. Further steps can be performed to enhance this process by validating the information in the recording w.r.t. their information on the form or scanned KYC data.

{\em Offline registration} is performed by certifiers that have physical location, e.g. banks, and assess the client face-to-face w.r.t. their KYC data. Then, they carry the registration online in the name of the client, who receives a private key at the end of the process. 
\end{remark}	

Registered clients can submit any KYC data \Doc=(\pk, \type, \data, \desc, \acc), by building a record \Entry defined as in Equation (\ref{eq: transaction}). As part of this process, a unique record identifier $\recid \sample \{0,1\}^\lambda$ is defined, and the unique symmetric keys are produced by 
$$\key_i \gets \PRF_{\seed} (\pk,\allowbreak \recid, i),$$ 
for $i \in \{1,2,3\}$. Then, the ciphertexts (\cdata, \cdesc, \cacc) are produced by calling the symmetric encryption scheme \SE with key $\key_1$ over \data, key $\key_2$ for \desc, and $\key_3$ for \acc. Finally, \DS is applied over (\type, \cdata, \cdesc, \cacc) to obtain signature \auth, and build record \Entry that is appended to the ledger \Ledger. 



Certification is done interactively between a client $\User(\pk_{\User},\sk_{\User})$ and a certifier $\Bank(\pk_{\Bank},\sk_{\Bank})$. Intuitively, the client provides the certifier with access to the KYC data and the certifier provides a signature over this set of KYC data. The authorization is performed through giving decryption keys for the encrypted KYC data that are logged in the ledger. First, the list of timestamped records $L$ for latest versions of the KYC data of client $\pk_{\User}$ that satisfy the policy \policy are taken from the ledger by both parties. Then, the client provides decryption keys for the data of all timestamped records in $L$ to \Bank. The certifier signs all KYC data, and sends the certificate to the client. If the timestamped records are invalid, i.e., their signature does not verify, or the client does not satisfy the policy, the certifier aborts.

Verification consists of checking that the certificate is valid w.r.t. KYC data provided. 
The list of KYC data can be provided by the client with the verifier additionally checking their timestamps against the ones in the ledger, or they can be extracted from the ledger by following Steps 1-8 from Figure \ref{fig: constr}. This would allow clients with a nomadic lifestyle to benefit from the same rights as all other clients. 

We allow clients to update KYC data they have submitted, by submitting new records. If clients do not have access to their initial KYC data, they can use their KYC data record identifier \recid, retrieve the timestamped record \Tran logged in the ledger, and extract the KYC data \Doc using decryption keys $\key_i \gets \PRF_{\seed} (\pk,\allowbreak \recid,i)$, for $i \in \{1,2,3\}$. Based on the values that client intends to update, i.e, the data \data, or list of certificates \desc, or access list \acc, the system would create a new KYC data $\Doc'$ where the required fields have been updated, and carry a new submission with $\Doc'$.

To satisfy {\em customer due diligence} requirements FIs need an in-depth knowledge of their clients to reason about their trustworthiness and likelihood of being involved in illegal activities. Changes in clients' lives, e.g., identity documents, occupation, address, etc, would be reflected by updates to the timestamps of the corresponding KYC data records in the ledger. FIs are able to get these updates through the same authorisation mechanism that they used to get access to client's KYC data during the onboarding phase.

\subsection{Security Analysis}

\begin{theorem}
The $\KYCS$ construction in Section \ref{sec: constr} offers data confidentiality, if $\;\SE$ is IND-CPA and $\PRF$ is a pseudo-random function.
\end{theorem}
\begin{proof}
Using unique keys derived via $\PRF$ to encrypt each KYC data together with the IND-CPA property for \SE provides sufficient guarantees for data confidentiality.
	
\textbf{Game $\Exp_0$:} We define experiment $\Exp_0$ as the data confidentiality experiment $Exp_{\adv,\KYCS}^{\privacy,\beta}(\upi,\lambda)$. Therefore, we trivially have
 	$$\prob{\Exp_0 = 1} = \prob{Exp_{\adv,\KYCS}^{\privacy,\beta}(\upi,\lambda) =1}.$$
 	
\textbf{Game $\Exp_1$:} We define experiment $\Exp_1$ as the experiment $\Exp_0$, except we replace 
$$\pcfor i \in \{1,2,3\} \pcdo \pcind[1] \key_i \gets \PRF_{\seed} (\pk,\allowbreak \recid, i),$$
from \Ocommit (line \ref{submit:keys} of \Commit) with 
$$\pcfor i \in \{1,2,3\} \pcdo \pcind[1] \key_i \gets \{0,1\}^\lambda.$$
This change should be undetected to the adversary, due to the pseudo-randomness property of the \PRF. Therefore, for $n$ \Ocommit oracle queries, it holds that
$$\abs{\prob{\Exp_0 = 1} - \prob{\Exp_1 = 1}} \leq 3n \times \advantage{\mathtt{prf}}{\bdv(\adv),\PRF}[].$$

\textbf{Reduction to IND-CPA Game:} Finally, we show that $\Exp_1$ can be reduced to IND-CPA of the symmetric encryption scheme. The IND-CPA adversary $\bdv(\adv)$ performs Steps 1-6 from game $\Exp_1$, except for Step 4 where he uses the encryption IND-CPA oracle query \texttt{Oenc($\cdot$)}. More precisely, he executes Steps 1-4 and 7-12 from \Commit, and replaces Step \ref{submit:keys}-6 with:
\begin{equation*}
\begin{split}
5:\hspace*{2mm} & \pcfor i \in \{2,3\} \pcdo \pcind[1] \key_i \gets \PRF_{\seed} (\pk,\allowbreak \recid, i)\\
6:\hspace*{2mm} & \cdata \gets \SE.\mathtt{Oenc}(\data)
\end{split}
\end{equation*}  
The key $\key_1$ is not generated, and the encryption of the KYC data is replaced with an oracle encryption. Moreover, \bdv(\adv) does not need to handle decryption queries as line ``\adv did not call \Oprove($\pk^*$, $\pk_{\Bank}$,$\Psi$) with $\Tran \in \Search(\policy,\pk^*)$'' makes these type of requests forbidden. The probability of $\adv$ to win game $\Exp_1$ is identical with the probability of \bdv(\adv) to win the IND-CPA experiment:
$$\prob{\Exp_1 = 1} = \prob{Exp_{\bdv(\adv),\SE}^{\mathtt{ind}-\mathtt{cpa}}(\lambda) =1}.$$	
As the advantage of an adversary is defined as the value greater than a random guess, i.e. 1/2, we have that the advantage of $\adv$ in $\Exp_1$ is $\abs{ \prob{\Exp_1 = 1} -1/2}$. Therefore, the following holds:
$$ \abs{ \prob{\Exp_1 = 1} -\cfrac{1}{2}} \leq \advantage{\mathtt{ind}-\mathtt{cpa}}{\bdv(\adv),\SE}[].$$	

The result of this theorem follow.
\end{proof}	

\begin{theorem}
The $\KYCS$ construction in Section \ref{sec: constr} offers certification compliance, if $\;\DS$ is EUF-CMA, $\SE$ is WKD, and $\;\PRF$ is pseudo-random function.
\end{theorem}
\begin{proof}
This security experiment measures the capabilities of an adversary to convince an honest certifier to issue an certificate either when the adversary is {\em impersonating an honest client that may satisfy the policy}, or when the adversary {\em doesn't satisfy the policy}. We split this experiment into two sub-experiments based on the winning condition of the adversary:
\begin{itemize}
\item $\mathtt{Exp}_1$: the adversary {\em impersonating an honest client}. This experiment is identical to $Exp_{\adv,\KYCS}^{\authentication}(\lambda)$, except that line 4.3 
is replaced with 
\begin{equation}\label{Exp1: return}
\begin{split}
4.3: \hspace*{10mm}           & \adv \mbox{ did not call } \Ocorrupt(\pk^*) \hspace*{16mm}
\end{split}
\end{equation}

\item $\mathtt{Exp}_2$: the adversary {\em doesn't satisfy the policy}. This experiment is identical to $\mathtt{Exp}_1$, except that line 4.3 from Eq. (\ref{Exp1: return}) is replaced with 
\begin{equation}\label{Exp2: return}
\begin{split}
4.3: \hspace*{8mm}           & \lnot \Psi(A) \mbox { for } A =\{ \Tran=(\pk^*,\cdot) | \\
           & \hspace*{4mm} \mbox{ for } \Tran \mbox{ added by } \Ocommit \mbox{ to } \Ledger \}	
\end{split}
\end{equation}
\end{itemize}

\textbf{Transition to $\mathtt{Exp}_1$ and $\mathtt{Exp}_2$:} The advantage of adversary $\adv$ is bounded by the above two probabilities, such that: 
$$
\advantage{\authentication}{\adv,\KYCS}[] \leq \prob{\mathtt{Exp}_1 =1} + \prob{\mathtt{Exp}_2 =1}.
$$ 

\textbf{Bound $\mathtt{Exp}_1$:} 
In this experiment the adversary wins if he can produce sufficient keys $K$ (Step 8 in Figure \ref{fig: constr}) and send them over an authenticated and confidential channel. The adversary can easily collect valid keys by acting the role of a certifier and running $\Oprove$ with the client $\User(\pk^*,\sk^*)$. However, he would not be able to send this set of keys over an authenticated channel, unless he can break the authentication property. Note that the definition does not consider Man-in-the-Middle attackers. Therefore, we have:
$$\prob{\mathtt{Exp}_1 = 1} \leq \advantage{\mathtt{auth}}{\bdv,\mathtt{Channel}}[].$$

\textbf{Bound $\mathtt{Exp}_2$:} 
Similar to experiment $\mathtt{Exp}_1$, the adversary needs to produce sufficient keys $K$ (Step 8 in Figure \ref{fig: constr}) that would be used to decrypt KYC data (Step 11 in Figure \ref{fig: constr}) that satisfies the certifier policy. However, the KYC data submitted by the adversary or client does not satisfy this policy. Therefore, the adversary can win if he submits one KYC data and can decrypt it to a different valid KYC data. 

\textit{Game $\Exp_1$:} The experiment $\Exp_1$ is identical to  $\mathtt{Exp}_2$, except the key derivation from \Ocommit (line \ref{submit:keys} of \Commit): 
$$\pcfor i \in \{1,2,3\} \pcdo \pcind[1] \key_i \gets \PRF_{\seed} (\pk,\allowbreak \recid, i),$$
is replaced by the following key generation 
$$\pcfor i \in \{1,2,3\} \pcdo \pcind[1] \key_i \gets \{0,1\}^\lambda.$$
The probability of the adversary to distinguish this change is bounded by the pseudo-randomness property of the \PRF, for $n$ \Ocommit oracle queries: 
$$\abs{\prob{\mathtt{Exp}_2 = 1} - \prob{\Exp_1 = 1}} \leq 3n \times \advantage{\mathtt{prf}}{\bdv(\adv),\PRF}[].$$

\textit{Bound on $\Exp_1$:} For adversary \adv\ to produce a different KYC data, he needs to compute a key $\key_1'$ that decrypts a cyphertext obtained by encrypting with a different key $\key_1 \neq \key_1'$. This reduces to the adversary breaking the WKD property of the symmetric encryption scheme.
$$ \prob{\Exp_1 = 1}  \leq \advantage{\mathtt{wkd}}{\bdv(\adv),\SE}[].$$	
 
The result of this theorem follows. 
\end{proof}	

%

\subsection{Practical Aspects in Building KYChain}

Typically, the on-boarding process is time consuming and costly both for financial institutions and their clients. These aspects can be even higher when taken together with inadequate handling of KYC data. KYChain intends to reduce these numbers significantly by performing the on-boarding a single time and re-using the certification from one FI to another. Furthermore, KYChain comes with a continuous monitoring system that allows FI to timely identify updates in their clients KYC data and request permission to view it.

The current estimates show FI spend 32 days/customer and upto \$20k/year for each client \cite{Marous}. These would only need to be supported at the first on-boarding for a client, as any additional on-boarding can be done by sharing the (certified) KYC data. 

\paragraph{Time for submitting KYC data} Clients submit encrypted KYC data to certifiers, that is then stored locally in a database. Then, a hash of the new database records are sent as a transaction to the public blockchain, i.e., Ethereum or Bitcoin. On average, Ethereum takes 15 seconds \cite{Ethereum-blocktime} to append a new block with multiple transactions, while Bitcoin takes 10 minutes \cite{Bitcoin-blocktime}. To consider a transaction irreversible, a number of additional blocks have to be added after the block that contains that transaction: 12-250 for Ethereum \cite{Ethereum-confirmation}, and 2-6 for Bitcoin \cite{Bitcoin-confirmation}. In the case of Ethereum, a KYC data transforms in a irreversible transaction in 10 minutes (15 seconds times 40 blocks). For Bitcoin, the time is longer at 40 minutes (10 minutes times 4 blocks).

\paragraph{Costs for submitting KYC data} Certifiers are supporting the transaction fees associated to submitting KYC data and certificates. Per block the Ethereum transaction fees is \$0.22 \cite{Ethereum-fees}, while the Bitcoin transaction fee is \$0.49 \cite{Bitcoin-fees}. Certifiers can optimize their costs by putting transactions together, and paying that block transaction fee per day, instead of per transaction.


\section{Conclusion}

In this paper we introduce KYChain, a privacy-preserving certification protocol for KYC data, that allows clients to securely share (certified) KYC data from one FI to another. To ensure integrity, we use a public blockchain with an external database to store encrypted (certified) KYC data. Clients have control over what KYC data is stored and with whom it is shared. Moreover, our system allows clients to update or remove any KYC data submitted. KYChain can significantly reduce time and costs with monitoring of changes to KYC data by the FIs through automated detection of updates in the ledger, and the on-boarding time by accepting certifications made by other FIs. We show that KYChain offers guarantees over confidentiality of the KYC data and authenticity of certification for KYC data compliance.  \\

\noindent \textbf{Acknowledgements.} Constantin C\u{a}t\u{a}lin Dr\u{a}gan and Mark Manulis were supported by the EPSRC project TAPESTRY (EP/N02799X).
%

\bibliographystyle{ACM-Reference-Format}
\bibliography{ref}

\end{document}